\documentclass[letterpaper, 10 pt, conference]{ieeeconf}  % Comment this line out
\IEEEoverridecommandlockouts                              % This command is only
\title{\LARGE \bf
Gradient Play in $n$-Cluster Games with Zero-Order Information}

\author{Tatiana Tatarenko, Jan Zimmermann, J\"urgen Adamy % <-this % stops a space
\thanks{The authors are with Control Methods and Robotics Lab at the TU Darmstadt, Germany.
        }
\thanks{}
}

\usepackage{amsmath, amssymb, epsfig, graphicx, euscript, bm}
\usepackage{cases}
\usepackage{float, subfigure}
\usepackage{color}
\usepackage{amsthm}
\usepackage{amsfonts}
\usepackage{cite, url}
\usepackage{breqn}
\usepackage{pgfplots}
\usepackage{tikz}
\usepackage[all,cmtip]{xy}
\usepackage{color,hyperref}
\definecolor{darkblue}{rgb}{0,0,1}
\hypersetup{colorlinks,breaklinks,
linkcolor=darkblue,urlcolor=darkblue,anchorcolor=darkblue,citecolor=darkblue}

\usepackage[normalem]{ulem} % to use \sout
\newcommand{\pro}[2]{{\mathcal{P}_{#1}\left\{{#2}\right\}}}

\newtheorem{theorem}{Theorem}
\newtheorem{definition}{Definition}
\newtheorem{proposition}{Proposition}

\newtheorem{remark}{Remark}

\newtheorem{assumption}{Assumption}

\newcommand{\B}{{\mathbb{B}}}
\newcommand{\Sph}{{\mathbb{S}}}

\newcommand{\R}{{\mathbb{R}}}

\newcommand{\s}{{\sigma}}

\newcommand{\E}{{\mathbb{E}}}

\newcommand{\A}{{\mathcal{A}}}

\newcommand{\Gra}{{\mathcal{G}}}

\newcommand{\tx}{{\tilde{x}}}

\newcommand{\Om}{{\Omega}}
\newcommand{\EuF}{{\EuScript F}}

\newcommand{\xlb}{\underline{x}}
\newcommand{\xub}{\overline{x}}
\newcommand{\bF}{{\mathbf{F}}}

\newcommand{\barx}{{\bar{x}}}

\begin{document}
%\abovedisplayskip = 0pt
%\belowdisplayskip = 0pt
%\abovedisplayshortskip = 0pt
%\belowdisplayshortskip= 0pt

\maketitle
\thispagestyle{empty}
\pagestyle{empty}

%%%%%%%%%%%%%%%%%%%%%%%%%%%%%%%%%%%%%%%%%%%%%%%%%%%%%%%%%%%%%%%%%%%%%%%%%%%%%%%%
\begin{abstract}                          % Abstract of not more than 200 words.
	We study a distributed approach for seeking a Nash equilibrium in $n$-cluster games with strictly monotone mappings. Each player within each cluster has access to the current value of her own smooth local cost function estimated by a zero-order oracle at some query point. We assume the agents to be able to communicate with their neighbors in the same cluster over some undirected graph. The goal of the agents in the cluster is to minimize their collective cost. This cost depends, however, on actions of agents from other clusters. Thus, a game between the clusters is to be solved. We present a distributed gradient play algorithm for determining a Nash equilibrium in this game. The algorithm takes into account the communication settings and zero-order information under consideration. We prove almost sure convergence of this algorithm to a Nash equilibrium given  appropriate estimations of the local cost functions' gradients.
\end{abstract}

%%%%%%%%%%%%%%%%%%%%%%%%%%%%%%%%%%%%%%%%%%%%%%%%%%%%%%%%%%%%%%%%%%%%%%%%%%%%%%%%
\section{Introduction}
%background
Distributed optimization and game theory provide powerful frameworks to deal with optimization problems arising in multi-agent systems. In generic distributed optimization problems, the cost functions of agents are distributed across the network, meaning that each agent has only partial information about the whole optimization problem which is to be solved. Game theoretic problems arise in such networks when the agents do not cooperate with each other and the cost functions of these non-cooperative agents are coupled by the  decisions of all agents in the system. The applications of game theoretic and distributed optimization approaches include, for example, electricity markets, power systems,  flow control problems and communication networks \cite{BasharSG, Scutaricdma, Nedich_Over}.

%motivation and definition of cluster games
On the other hand, cooperation and competition coexists in many practical situations, such as cloud computing, hierarchical optimization in Smart Grid, and adversarial networks \cite{Cortes2013, Jarrah2015, Niyato2011}. A body of recent work has been devoted to analysis of non-cooperative games and distributed optimization problems in terms of a single model called \emph{$n$-cluster games} \cite{Pang2020, Zeng2019, Ye2018, Ye2020, Zimmermann2021, Meng2020}.
In such $n$-cluster games, each cluster corresponds to a player whose goal is to minimize her own cost function.
However, the clusters in this game are not the actual decision-makers as the optimization of the cluster's objective is controlled by the agents belonging to the corresponding cluster. Each of such agents has her own local cost function, which is available only to this agent, but depends on the joint actions of agents in all clusters.
The cluster's objective, in turn, is the sum of the local cost functions of the agents within the cluster. Therefore, in such models, each agent intends to find a strategy to achieve a Nash equilibrium in the resulting $n$-cluster game, which is a stable state that minimizes the cluster's cost functions in response to the actions of the agents from other clusters.

%recent work
Continuous time algorithms for the distributed Nash equilibria seeking problem in multi-cluster games were proposed in \cite{Zeng2019, Ye2018, Ye2020}.
The paper \cite{Ye2018} solves an unconstrained multi-cluster game by using gradient-based algorithms, whereas the works \cite{Ye2020} and \cite{Zeng2019} propose a gradient-free algorithm, based on zero-order information, for seeking Nash and generalized Nash equilibria respectively.
In discrete time domain, the work \cite{Meng2020} presents a leader-follower based algorithm, which can solve unconstrained multi-cluster games in linear time.
The authors in \cite{Zimmermann2021} extend this result to the case of leaderless architecture. Both papers \cite{Meng2020, Zimmermann2021} prove linear convergence in games with strongly monotone mappings and first-order information, meaning that agents can calculate gradients of their cost functions and use this information to update their states.
In contrast to that, the work \cite{Pang2020} deals with a gradient-free approach to the cluster games. However, the gradient estimations are constructed in such a way that only convergence to a neighborhood of the equilibrium can be guaranteed. Moreover, these estimations are obtained by using two query points, for which an extra coordination between the agents is required.

%contributions
Motivated by relevancy of $n$-cluster game models in many engineering applications, we present a discrete time distributed procedure to seek Nash equilibria in $n$-cluster games with zero-order information.
We consider settings, where agents can communicate with their direct neighbors within the corresponding cluster over some undirected graph.
However, in many practical situations the agents do not know the functional form of their objectives and can only access the current values of their objective functions at some query point. Such situations arise, for example, in electricity markets with unknown price functions \cite{elmark}. In such cases, the information structure is referred to as \textit{zero-order oracle}. Our work focuses on zero-order oracle information settings and, thus, assumes agents to have no access to the analytical form of their cost functions and gradients.
The agents instead construct their local query points and get the corresponding cost values from the oracle. Based on these values, the agents estimate their local gradients to be able to follow the step in the gradient play procedure. We formulate the sufficient conditions and provide some concrete example on how to estimate the gradients to guarantee the almost sure convergence of the resulting algorithm to Nash equilibria in $n$-cluster games with  strictly monotone game mappings. To the best of our knowledge, we present the first algorithm solving $n$-cluster games with zero-order oracle and the corresponding one-point gradient estimations.

%organization
The paper is organized as follows. In Section~\ref{sec:problem} we formulated the $n$-cluster game with undirected communication topology in each cluster and zero-order oracle information.
Section~\ref{sec:main} introduces the gradient play algorithm which is based on the one-point gradient estimations. The convergence result is presented in  Section~\ref{sec:main} as well. Section~\ref{sec:gradSampl} provides an example of query points and gradient estimations which guarantee convergence of the algorithm discussed in Section~\ref{sec:main}. Section~\ref{sec:sim} presents some simulation results. Finally, Section~\ref{sec:conclusion} concludes the paper.

\textbf{Notations.}
The set $\{1,\ldots,n\}$ is denoted by $[n]$.
For any function $f:K\to\R$, $K\subseteq\R^n$, $\nabla_i f(x) = \frac{\partial f(x)}{\partial x_i}$ is the partial derivative taken in respect to the $i$th coordinate of the vector variable $x\in\R^n$.
We consider real normed space $E$, which is the space of real vectors, i.e. $E = \R^n$.  We use $(u,v)$ to denote the inner product in $E$.
We use $\|\cdot\|$ to denote the Euclidean norm induced by the standard dot product in $E$. Any mapping $g:E\to E$ is said to be \emph{strictly monotone on $Q\subseteq E$}, if $(g(u)-g(v), u - v) >0$ for any $u,v\in Q$, where $u\ne v$.
We use $\B_r(p)$ to denote the ball of the radius $r\ge 0$ and the center $p\in E$ and $\Sph$ to denote the unit sphere with the center in $0\in E$.
We use $\pro{\Om}{v}$ to denote the projection of $v\in E$ to a set $\Om\subseteq E$.
The mathematical expectation of a random value $\xi$ is denoted by $\E\{\xi\}$.
We use the big-$O$ notation, that is, the function $f(x): \R\to\R$ is $O(g(x))$ as $x\to a$, $f(x)$ = $O(g(x))$ as $x\to a$, if $\lim_{x\to a}\frac{|f(x)|}{|g(x)|}\le K$ for some positive constant $K$.

\section{Nash Equilibrium Seeking}\label{sec:problem}
\subsection{Problem Formulation}
We consider a non-cooperative game between $n$ clusters. Each cluster $i\in[n]$ itself consists of $n_i$ agents. Let $J^j_i$ and $\Om^j_i\subseteq \R$\footnote{All results below are applicable for games with different dimensions $\{d^j_i\}$ of the action sets $\{\Om^j_i\}$. The one-dimensional case is considered for the sake of notation simplicity.} denote respectively the cost function and the feasible action set of the agent $j$ in the cluster $i$. We denote the joint action set of the agents in the cluster $i$ by $\Om_i = \Om^1_i\times\ldots\times\Om_i^{n_i}$. Each function $J^j_i(x_i,x_{-i})$, $i\in[n]$, depends on $x_i = (x_i^1,\ldots, x_i^{n_i}) \in\Om_i$, which represents the joint action of the agents within the cluster $i$, and $x_{-i} \in\Om_{-i}=\Om_1\times\ldots\times\Om_{i-1}\times\Om_{i+1}\times\Om_n$,  denoting the joint action of the agents from all clusters except for the cluster $i$.
The cooperative cost function in the cluster $i\in[n]$ is, thus, $J_i(x_i,x_{-i}) = \frac{1}{n_i}\sum_{j=1}^{n_i}J^j_i(x_i,x_{-i})$.\\
 We assume that the agents within each cluster can interact over an undirected communication graph $\Gra_i([n_i],\A_i)$. The set of nodes is the set of the agents $[n_i]$ and the set of undirected arcs $\A_i$ is such that $(k,j)\in\A_i$ if and only if $(j,k)\in\A_i$, i.e. there is a bidirectional communication link between $k$ to $j$, over which information in form of a message can be sent from the agent $k$ to the agent $j$ and vice versa in the cluster $i$.\\
 However, there is \emph{no explicit communication between the clusters}. Instead, we consider the following \emph{zero-order} information structure in the system: No agent has access to the analytical form of any cost function, including its own. Each agent can only observe the value of its local cost function given any joint action of all agents in the system. Formally, given a joint action $x\in\Om$, each agent $j\in[n_i]$, $i\in[n]$ receives the value $J^j_i(x)$ from a zero-order oracle. Especially, no agent has or receives any information about the gradient.

Let us denote the game between the clusters introduced above by $\Gamma(n,\{J_i\},\{\Om_i\}, \{\Gra_i\})$.
We make the following assumptions regarding the game $\Gamma$:
\begin{assumption}\label{assum:convex}
 The $n$-cluster game under consideration is \emph{strictly convex}. Namely, for all $i\in[n]$, the set $\Om_i$ is convex, the cost function $J_i(x_i, x_{-i})$ is continuously differentiable in $x_i$ for each fixed $x_{-i}$. Moreover, the game mapping, which is defined as
  \begin{align}\label{eq:gamemapping}
&\bF(x)\triangleq\left[\nabla_1 J_1(x_1,x_{-1}), \ldots, \nabla_n J_n(x_n,x_{-n})\right]^T
%\bF(x)&\in\R^{n}.
 \end{align}
  is \emph{strictly monotone} on $\Om$.
\end{assumption}

\begin{assumption}\label{assum:Lipschitz}
  Each function $\nabla_{i}J_i^j(x_i,x_{-i})$ is Lipschitz continuous on $\Om$.
\end{assumption}

\begin{assumption}\label{assum:compact}
  The action sets $\Om_j^i$, $j\in[n_i]$, $i\in[n]$, are compact. Moreover, for each $i$ there exists a so called safety ball $\B_r(p)\subseteq \Om_i$ with $r_i>0$ and $p_i\in\Om_i$\footnote{Existence of the safety ball is required to construct feasible points for costs' gradient estimations in the zero-order settings under consideration (see \cite{Bravo}).}.
\end{assumption}

The assumptions above are standard in the literature on both game-theoretic and zero-order optimization \cite{Bravo}.
Finally, we make the following assumption on the communication graph, which guarantees sufficient information  "mixing" in the network within each cluster.
\begin{assumption}\label{assum:connected}
The underlying undirected communication graph $\Gra_i([n_i],\A_i)$ is connected for all $i=1,\ldots, n$. The associated non-negative mixing matrix $W_i=[w^i_{kj}]\in\R^{n\times n}$ defines the weights on the undirected arcs such that $w^i_{kj}>0$ if and only if $(k,j)\in\A_i$ and $\sum_{k=1}^{n_i}w^i_{kj} = 1$, $\forall k\in[n_i]$.
\end{assumption}

One of the stable solutions in any game $\Gamma$ corresponds to a Nash equilibrium defined below.
\begin{definition}\label{def:NE}
 A vector $x^*=[x_1^*,x_2^*,\cdots, x_n^*]^T\in\Om$ is called a \emph{Nash equilibrium} if for any $i\in[n]$ and $x_i\in \Om_i$
 $$J_i(x_i^*,x_{-i}^*)\le J_i(x_{i},x_{-i}^*).$$
 \end{definition}
In this work, we are interested in \emph{distributed seeking of a Nash equilibrium} in any game $\Gamma(n,\{J_i\},\{\Om_i\},\{\Gra_i\})$ with the information structure described above and for which Assumptions~\ref{assum:convex}-\ref{assum:connected} hold.

\subsection{Existence and Uniqueness of the Nash Equilibrium}
In this subsection, we demonstrate the existence of the Nash equilibrium for $\Gamma(n,\{J_i\},\{\Om_i\},\{\Gra_i\})$ under Assumptions~\ref{assum:convex} and \ref{assum:compact}. For this purpose we recall the results connecting Nash equilibria and solutions of variational inequalities from \cite{FaccPang1}.

\begin{definition}\label{def:VI}
Consider a set $Q \subseteq \R^d$ and a mapping $g$: $Q \to \R^d$. A
\emph{solution $SOL(Q,g)$ to the variational inequality problem} $VI(Q,g)$ is a set of vectors $q^* \in Q$ such that $\langle g(q^*), q-q^*\rangle \ge 0$, for any $q \in Q$.
\end{definition}

The following theorem is the well-known result on the connection between Nash equilibria in games and solutions of a definite variational inequality (see Corollary 1.4.2 in \cite{FaccPang1}).

\begin{theorem}\label{th:VINE}
 Consider a non-cooperative game $\Gamma$. Suppose that the action sets of the players $\{\Om_i\}$ are closed and convex, the cost functions $\{J_i(x_i,x_{-i})\}$ are continuously differentiable and convex in $x_i$ for every fixed $x_{-i}$ on the interior of the joint action set $\Om$. Then, some vector $x^*\in \Om$ is a Nash equilibrium in $\Gamma$, if and only if $x^*\in SOL(\Om,\bF)$, where $\bF$ is the game mapping defined by \eqref{eq:gamemapping}.
\end{theorem}

Next, we formulate the result guaranteeing existence and uniqueness of $SOL(Q,g)$ in the case of strictly monotone map $Q$ (see Corollary 2.2.5 and Proposition 2.3.3 in \cite{FaccPang1}).
\begin{theorem}\label{th:existVI}
 Given the $VI(Q,g)$, suppose that $Q$ is compact and the mapping $g$ is strictly monotone. Then, the solution $SOL(Q,g)$ exists and is a singleton.
\end{theorem}

Taking into account Theorems~\ref{th:VINE} and \ref{th:existVI}, we obtain the following result.
\begin{theorem}\label{th:exist}
   Let $\Gamma(n,\{J_i\},\{\Om_i\},\{\Gra_i\})$ be a game for which Assumptions~\ref{assum:convex} and \ref{assum:compact} hold. Then, there exists the unique Nash equilibrium in $\Gamma$. Moreover, the Nash equilibrium in $\Gamma$ is the solution of $VI(\Om,\bF)$, where $\bF$ is the game mapping (see \eqref{eq:gamemapping}).
\end{theorem}

Thus, if Assumptions~\ref{assum:convex} and \ref{assum:compact} hold, we can guarantee existence and uniqueness of the Nash equilibrium in the game $\Gamma(n,\{J_i\},\{\Om_i\},\{\Gra_i\})$ under consideration and use the corresponding variational inequality in the analysis of the optimization procedure presented below.

\section{Main Results}\label{sec:main}
\subsection{Zero-order gradient play between clusters}
To deal with the zero-order information available to the agents and local state exchanges within the clusters, we assume each agent $j$ from the cluster $i$ maintains a \emph{local variable}
\begin{align}\label{eq:vector}
x^{(j)}_{i}=[x^{(j)1}_i,\cdots,x^{(j)j-1}_i,x^j_i,x^{(j)i+1}_i,\cdots,x^{(j)n_i}_i]^T\in\Om_i,
\end{align}
which is her estimation of the joint action $x_i=[x^1_i,x^2_i,\cdots,x^{n_i}_i]^T\in\Om_i$ of the agents from her cluster $i$.
Here, $x^{(j)k}_i\in\Om^j_i$ is player $k$'s estimate of $x^j_i$ and  $x^{(j)j}_i=x^j_i\in\Om^j_i$ is the action of agent $j$ from cluster $i$.
The goal of the agents within each cluster is to update their local variables in such a way that the joint action $x = (x_1,\ldots, x_n)\in\Om$ with $x_i = (x^1_i, \ldots, x^{n_i}_i) \in\Om_i$ converges to the Nash equilibrium in the game $\Gamma$ between the clusters as time runs.
To let the agents achieve this goal, we aim to adapt the standard projected gradient play approach to the cluster game with the zero-order information.

At this point we assume each agent $j\in[n_i]$, $i\in[n]$, based on its local estimation $x^{(j)}_i$, constructs a feasible query point $\hat x^{(j)}_i\in\Om_i$ and sends it to the oracle. As a reply from the oracle, the agent receives the value $J^j_i(\hat x^{(j)}_i,\hat{\tx}_{-i})$.
The vector $\hat{\tx}_{-i}$ here corresponds to the point obtained by some combination of the query vectors sent by the agents from the other clusters. Formally,
\begin{align}\label{eq:tx1}
\hat{\tx}_{-i} = (\hat x^{(j_1)}_1,\ldots,\hat x^{(j_{i-1})}_{i-1},\hat x^{(j_{i+1})}_{i+1},\ldots,\hat x^{(j_n)}_n),
\end{align}
where $j_k$ denotes some agent from the cluster $k\in[n]$, $k\ne i$.
Further each agent $j\in[n_i]$, $i\in[n]$, uses the received value $J^j_i(\hat x^{(j)}_i,\hat{\tx}_{-i})$ to obtain the random estimation $d^j_i$ of her local cost's gradient $\nabla_iJ^j_i$ at the point $(x^{(j)}_i,\tx_{-i})$, where
\begin{align}\label{eq:tx}
    \tx_{-i} = (x^{(j_1)}_1,\ldots,x^{(j_{i-1})}_{i-1},x^{(j_{i+1})}_{i+1},\ldots,x^{(j_n)}_n)
\end{align}
corresponds to the local estimations of other agents (one for each cluster different from $i$) based on which query points are obtained.
Thus, $d^j_i = d^j_i(J^j_i(\hat x^{(j)}_i,\hat{\tx}_{-i}))\in\R^{n_i}$.
As $d^j_i$ is an estimation of $\nabla_iJ^j_i(x^{(j)}_i,\tx_{-i})$, we represent this vector by the following decomposition:
\begin{align}\label{eq:gradEst}
  d^j_i = \nabla_iJ^j_i(x^{(j)}_i,\tx_{-i}) + e^j_i,
\end{align}
where $e^j_i$ is a random vector reflecting inaccuracy of the obtained estimation, i.e. the estimation error vector.
Note that for the joint query point $(\hat x^{(j)}_i, \hat {\tx}_{-i})$ the oracle is free to choose any combination $\hat {\tx}_{-i}$ of the local queries defined in \eqref{eq:tx1}.

Now we are ready to formulate the gradient play between the clusters.
Starting with an arbitrary $x^{(j)}_{i}(0)\in\Om_i$, each agent $j$ updates the local estimation vector $x^{(j)}_{i}$, $j\in[n_i]$, $i\in[n]$, as follows:
\begin{align}\label{eq:pbalg}
  x^{(j)}_{i}(t+1) = \pro{\Om_i}{\sum_{l=1}^{n_i}w^i_{jl}x^{(l)}_i(t) - \alpha_t d^j_i(t)},
\end{align}
where the time-dependent parameter $\alpha_t>0$ corresponds to the step size.

Let $\EuF_t$ be the $\sigma$-algebra generated by the estimations $\{x^{(j)}_{i}(m)\}_{m=0}^t$ up to time $t$, $j\in[n_i]$, $i\in[n]$.
Let $\barx_i(t)=\frac{1}{n_i}\sum_{j=1}^{n_i}x^{(j)}_{i}(t)$ be the running average of the agents' estimations vectors within the cluster $i$. The following proposition describes the behavior of $\barx_i(t)$ in the long run.

\begin{proposition}\label{prop:runav}
  Let Assumptions~\ref{assum:compact} and \ref{assum:connected} hold and $x^{(j)}_i(t)$, $j\in[n_i]$, $i\in[n]$, be updated according to \eqref{eq:pbalg}.  Then for all $j\in[n_i]$, $i\in[n]$
  \begin{enumerate}
    \item if $\lim_{t\to\infty}\alpha_t = 0$ and $\lim_{t\to\infty}\alpha_t\sqrt{\E\{\|e^j_i(t)\|^2|\EuF_t\}} = 0$ almost surely, then $\lim_{t\to\infty}\|x^{(j)}_i(t) - \barx_i(t)\|= 0$ almost surely;
    \item if $\sum_{t=0}^{\infty}\alpha^2_t <\infty$ and $\sum_{t=0}^{\infty}\alpha^2_t{\E\{\|e^j_i(t)\|^2|\EuF_t\}} <\infty$ almost surely, then
    $\sum_{t=0}^{\infty}\alpha_t\|x^{(j)}_i(t) - \barx_i(t)\|<\infty$.
  \end{enumerate}
\end{proposition}
\begin{proof}
	Follows from Lemma 8 in \cite{Nedich_projected}\footnote{The proof can be repeated up to (37) in \cite{Nedich_projected}. The inequality (37) and the analysis afterward stay valid in terms of the conditional expectation $\E\{\cdot|\EuF_t\}$.}.
\end{proof}

In view of the proposition above and to be able to analyze behavior of the algorithm by means of the running averages $\barx_i(t)$, $i\in[n]$ we make the following assumption on the balance between the step size $\alpha_t$ and the error term $e^j_i(t)$.

\begin{assumption}\label{assum:step}
The step size $\alpha_t$ and the error term $e^j_i(t)$ are such that
  \begin{align*}
     &\sum_{t=0}^{\infty}\alpha_t = \infty, \, \sum_{t=0}^{\infty}\alpha^2_t < \infty, \\
     & \sum_{t=0}^{\infty}\alpha_t\E\{\|e^j_i((t))\||\EuF_t\}<\infty \, \mbox{ almost surely}, \\ &\sum_{t=0}^{\infty}\alpha^2_t\E\{\|e^j_i((t))\|^2|\EuF_t\}<\infty  \, \mbox{ almost surely}.
  \end{align*}
\end{assumption}
In Section~\ref{sec:gradSampl} we shed light on how the gradients can be sampled to guarantee fulfillment of Assumption~\ref{assum:step}.
With Proposition~\ref{prop:runav} in place, we are ready to prove the main result formulated in the theorem below.

\begin{theorem}
  Let Assumptions~\ref{assum:convex}-\ref{assum:step} hold and $x^{(j)}_i(t)$, $j\in[n_i]$, $i\in[n]$, be updated according to \eqref{eq:pbalg}.
  Then the joint action $x(t) = (x_1(t),\ldots,x_n(t))$ converges almost surely to the unique Nash equilibrium $x^*$ in the $\Gamma(n,\{J_i\},\{\Om_i\},\{\Gra_i\})$, i.e.
  $\Pr\{\lim_{t\to\infty}\|x(t)-x^*\|=0\}=1$.
\end{theorem}
\begin{proof}
  Let $x^*=(x_1^*,\ldots, x_n^*)$ be the unique Nash equilibrium in the game $\Gamma(n,\{J_i\},\{\Om_i\},\{\Gra_i\})$, see Theorem~\ref{th:exist}. We proceed with estimating the distance between $x^{(j)}_i(t+1)$ and $x_i^*$. Let $v^j_i(t)=\sum_{l=1}^{n_i}w^i_{jl}x^{(l)}_i(t)$. As $x_i^*\in\Om_i$, we can use the non-expansion of the projection operator to conclude that almost surely (a.s.)\footnote{In the following discussion the big-$O$ notation is defined under the limit $t\to\infty$ (see \textbf{Notations}).}
  \begin{align}\label{eq:proof1}
    &\|x^{(j)}_i(t+1)- x_i^*\|^2  \le \| v^j_i(t)- \alpha_t d^j_i(t) -x_i^*\|^2 \cr
     & = \|v^j_i(t)-x_i^*\|^2\cr
      &\quad- 2\alpha_t(d^j_i(t), v^j_i(t)-x_i^*) + \alpha_t^2\|d^j_i(t)\|^2 \cr
     & =\|v^j_i(t)-x_i^*\|^2 \cr
     &\quad- 2\alpha_t(\nabla_iJ^j_i(x^{(j)}_i(t),\tx_{-i}(t)), v^j_i(t)-x_i^*)\cr
     &\quad - 2\alpha_t(e^j_i((t)), v^j_i(t)-x_i^*) + O(\alpha_t^2(1+\|e^j_i(t)\|^2))\cr
     &\le \|v^j_i(t)-x_i^*\|^2 \cr
     &\quad- 2\alpha_t(\nabla_iJ^j_i(x^{(j)}_i(t),\tx_{-i}(t)), v^j_i(t)-x_i^*)\cr
     &\quad + O(\alpha_t\|e^j_i((t))\|) + O(\alpha_t^2(1+\|e^j_i(t)\|^2)),
  \end{align}
  where in the last equality we used \eqref{eq:gradEst}, which implies that a.s.
   \[\|d^j_i(t)\|^2 \le 2(\|\nabla_iJ^j_i(x^{(j)}_i(t),\tx_{-i}(t))\|^2+\|e^j_i(t)\|^2)\]
   and, thus, $\|d^j_i(t)\|^2 = O(1+\|e^j_i(t)\|^2)$ a.s. (see Assumptions~\ref{assum:convex} and~\ref{assum:compact}), whereas in the last inequality we used the Cauchy–Schwarz inequality, implying
  \[- (e^j_i((t)), v^j_i(t)-x_i^*)\le \|e^j_i((t))\|v^j_i(t)-x_i^*\| \, \mbox{a.s.},\]
  and Assumption~\ref{assum:compact} implying almost sure boundedness of $\|v^j_i(t)-x_i^*\|$.
  We focus now on the terms $\|v^j_i(t)-x_i^*\|^2$ and $- 2\alpha_t(\nabla_iJ^j_i(x^{(j)}_i(t),\tx_{-i}(t)), v^j_i(t)-x_i^*)$.
  Due to Assumption~\ref{assum:connected}, we have that a.s.
  \[\|v^j_i(t)-x_i^*\|^2 = \|\sum_{l=1}^{n_i}w^i_{jl}x^{(l)}_i(t)-x_i^*\|^2\le\sum_{l=1}^{n_i}w^i_{jl}\|x^{(l)}_i(t)-x_i^*\|^2.\]
  And, as $\sum_{j=1}^{n_i}w^i_{jl}=1$, we obtain that a.s.
  \begin{align}\label{eq:sumv}
    \sum_{j=1}^{n_i}\|v^j_i(t)-x_i^*\|^2&\le \sum_{l=1}^{n_i}(\sum_{j=1}^{n_i}w^i_{jl})\|x^{(l)}_i(t)-x_i^*\|^2\cr
     &= \sum_{l=1}^{n_i}\|x^{(l)}_i(t)-x_i^*\|^2.
  \end{align}
  Next,
  \begin{align}\label{eq:sumscpr}
    (\nabla_iJ^j_i(x^{(j)}_i(t),&\tx_{-i}(t)), v^j_i(t)-x_i^*)\cr
     =&(\nabla_iJ^j_i(x^{(j)}_i(t),\tx_{-i}(t)), v^j_i(t)-x_i^*)\cr
    & - (\nabla_iJ^j_i(x^{(j)}_i(t),\tx_{-i}(t)), \barx_i(t)-x_i^*)\cr
    & + (\nabla_iJ^j_i(x^{(j)}_i(t),\tx_{-i}(t)), \barx_i(t)-x_i^*)\cr
    & - (\nabla_iJ^j_i(\barx_i(t)\barx_{-i}(t)), \barx_i(t)-x_i^*)\cr
    & + (\nabla_iJ^j_i(\barx_i(t)\barx_{-i}(t)), \barx_i(t)-x_i^*),
  \end{align}
  where $\bar\tx_{-i}(t)\in\R^{\sum_{k\ne i}n_k}$ is the joint running average of the agents' local variable over all clusters except for the cluster $i$ (see more details in \eqref{eq:tx}).
  Thus, by applying the Cauchy–Schwarz inequality to \eqref{eq:sumscpr}, we get
  \begin{align}\label{eq:sumscpr1}
  &-(\nabla_iJ^j_i(x^{(j)}_i(t),\tx_{-i}(t)), v^j_i(t)-x_i^*) \cr
  &\le \|\nabla_iJ^j_i(x^{(j)}_i(t),\tx_{-i}(t))\| \|v^j_i(t)-\barx_i(t)\|\cr
  &\quad +\|\nabla_iJ^j_i(x^{(j)}_i(t),\tx_{-i}(t))- \nabla_iJ^j_i(\barx_i(t)\barx_{-i}(t))\|\cr
  &\qquad\qquad\qquad\qquad\qquad\qquad\times \|\barx_i(t)-x_i^*\|\cr
  &\quad - (\nabla_iJ^j_i(\barx_i(t)\barx_{-i}(t)), \barx_i(t)-x_i^*), \, \mbox{a.s.}.
  \end{align}
  Taking into account almost sure boundedness of $\|\nabla_iJ^j_i(x^{(j)}_i(t),\tx_{-i}(t))\|$ and $\|\barx_i(t)-x_i^*\|$ (see Assumptions~\ref{assum:convex} and ~\ref{assum:compact}) and Assumption~\ref{assum:Lipschitz}, we conclude that
  \begin{align}\label{eq:sumscpr2}
  &-(\nabla_iJ^j_i(x^{(j)}_i(t),\tx_{-i}(t)), v^j_i(t)-x_i^*) \cr
  &\le O(\|v^j_i(t)-\barx_i(t)\|)\cr
  &\quad +O(\|x^{(j)}_i(t) - \barx_i(t)\| + \|\tx_{-i}(t) - \barx_{-i}(t)\|)\cr
  &\quad - (\nabla_iJ^j_i(\barx_i(t),\tx_{-i}(t)), \barx_i(t)-x_i^*).
  \end{align}
  Thus, we get from \eqref{eq:proof1}
  \begin{align}\label{eq:proof2}
    &\|x^{(j)}_i(t+1)- x_i^*\|^2  \le \|v^j_i(t)-x_i^*\|^2 \cr
     &\quad- 2\alpha_t(\nabla_iJ^j_i(\barx_i(t)\barx_{-i}(t)), \barx_i(t)-x_i^*)\cr
     &\quad + 2\alpha_tO(\|v^j_i(t)-\barx_i(t)\|)\cr
     &\quad + 2\alpha_tO(\|x^{(j)}_i(t) - \barx_i(t)\| + \|\tx_{-i}(t) - \barx_{-i}(t)\|)\cr
     &\quad + O(\alpha_t\|e^j_i((t))\|) + O(\alpha_t^2(1+\|e^j_i(t)\|^2)), \,\mbox{a.s.}.
  \end{align}
  Analogously to \eqref{eq:sumv}
  \[\sum_{j=1}^{n_i}\|v^j_i(t)-\barx_i(t)\| \le \sum_{l=1}^{n_i} \||x^{(j)}_i(t)-\barx_i(t)\|.\]
  Therefore, by averaging both sides of \eqref{eq:proof2} over $j=1,\ldots,n_i$ and taking the conditional expectation in respect to $\EuF_t$ (below we use the notation $\E\{\cdot|\EuF_t\} = \E_t\{\cdot\}$), we obtain that a.s.
    \begin{align}\label{eq:proof3}
    &\frac{1}{n_i}\sum_{j=1}^{n_i}\E_t\{\|x^{(j)}_i(t+1)- x_i^*\|^2\}  \le \frac{1}{n_i}\sum_{j=1}^{n_i}\|x^{(j)}_i(t)-x_i^*\|^2 \cr
     &\quad- 2\alpha_t\frac{1}{n_i}\sum_{j=1}^{n_i}(\nabla_iJ^j_i(\barx_i(t)\barx_{-i}(t)), \barx_i(t)-x_i^*)\cr
     &\quad + O\left(\frac{1}{n_i}\sum_{j=1}^{n_i}\alpha_t\|x^{(j)}_i(t)-\barx_i(t)\|\right)\cr
     &\quad + O(\alpha_t\|\tx_{-i}(t) - \barx_{-i}(t)\|)\cr
     &\quad + \frac{1}{n_i}\sum_{j=1}^{n_i}\left(O(\alpha_t\E_t\{\|e^j_i((t))\|\}+\alpha_t^2(1+\E_t\{\|e^j_i(t)\|^2\}))\right)\cr
    &= \frac{1}{n_i}\sum_{j=1}^{n_i}\|x^{(j)}_i(t)-x_i^*\|^2 \cr
     &\quad- 2\alpha_t(\nabla_iJ_i(\barx_i(t)\barx_{-i}(t)), \barx_i(t)-x_i^*)+ h_i(t),
  \end{align}
  where
  \begin{align*}
    h_i(t) & = O\left(\frac{1}{n_i}\sum_{j=1}^{n_i}\alpha_t\|x^{(j)}_i(t)-\barx_i(t)\|\right)\cr
         &\quad + O(\alpha_t\|\tx_{-i}(t) - \barx_{-i}(t)\|)\cr
     &\quad + \frac{1}{n_i}\sum_{j=1}^{n_i}O(\alpha_t\E_t\{\|e^j_i((t))\|\})\cr
     &\quad+\frac{1}{n_i}\sum_{j=1}^{n_i}O(\alpha_t^2(1+\E_t\{\|e^j_i(t)\|^2\})).
  \end{align*}
  By taking into account Proposition~\ref{prop:runav} 2) and the definition of $\tx_{-i}(t)$ (see \eqref{eq:tx}), we conclude that a.s.
   \[\sum_{t=0}^{\infty}O\left(\frac{1}{n_i}\sum_{j=1}^{n_i}\alpha_t\|x^{(j)}_i(t)-\barx_i(t)\|\right)<\infty,\]
         \[\sum_{t=0}^{\infty} O(\alpha_t\|\tx_{-i}(t) - \barx_{-i}(t)\|)<\infty.\]
         Moreover, due to Assumption~\ref{assum:step},
   \[\sum_{t=0}^{\infty}\frac{1}{n_i}\sum_{j=1}^{n_i}O(\alpha_t\E_t\{\|e^j_i((t))\|\}<\infty,\]
         \[\sum_{t=0}^{\infty} \frac{1}{n_i}\sum_{j=1}^{n_i}\alpha_t^2(1+\E_t\{\|e^j_i(t)\|^2\}))<\infty\]
         almost surely. Thus,
  \begin{align}\label{eq:h}
    \sum_{t=0}^{\infty} h_i(t)<\infty \, \mbox{a.s. for all }i\in[n].
  \end{align}
  Next, let us introduce the vector $u(t) = (u_1(t),\ldots,u_n(t)),$ where $u_i = \left(\frac{1}{n_i}\sum_{j=1}^{n_i}\|x^{(j)}_i(t)-x_i^*\|^2\right)^{\frac{1}{2}}$. Therefore, summing \eqref{eq:proof3} over $i\in[n]$ implies
  \begin{align}\label{eq:proof4}
    &\E_t\|u(t+1)\|^2\le \|u(t)\|^2 - 2\alpha_t(\bF(\barx(t)), \barx(t)-x^*)\cr
    &\qquad\qquad\qquad\qquad\qquad+ \sum_{i=1}^{n}h_i(t)\cr
    &\le \|v(t)\|^2 - 2\alpha_t(\bF(\barx(t))-\bF(x^*), \barx(t)-x^*)\cr
    &\qquad\qquad\qquad\qquad\qquad+\sum_{i=1}^{n}h_i(t),
  \end{align}
  where in the last inequality we used the fact that $x^*$ is the Nash equilibrium in $\Gamma(n,\{J_i\},\{\Om_i\}, \{\Gra_i\})$ and, thus,
  $(\bF(x^*), \barx(t)-x^*)\ge 0$ a.s. for all $t$ (see Theorem~\ref{th:VINE}). Due to the strictly monotone mapping (see Assumption~\ref{assum:convex}), which implies
   \[(\bF(\barx(t))-\bF(x^*), \barx(t)-x^*)\ge 0,\]
   and \eqref{eq:h}, we can apply the Robbins and Siegmund result (see Theorem~\ref{th:th_nonnegrv} in Appendix) to the inequality \eqref{eq:proof4}. With that, we conclude that $\|u(t)\|^2$ converges a.s. as $t\to\infty$ and
  \[\sum_{t=1}^{\infty}\alpha_t(\bF(\barx(t))-\bF(x^*), \barx(t)-x^*)<\infty \, \mbox{a.s.}\]
  Taking the inequality above and the fact that $\sum_{t=0}^{\infty}\alpha_t = \infty$ into account, we conclude that
  \[\liminf_{t\to\infty}(\bF(\barx(t))-\bF(x^*), \barx(t)-x^*)=0 \,\mbox{a.s.},\]
  which together with strict monotonicity of $\bF$  implies existence of the subsequence $\barx(t_m)$ such that
  $\lim_{m\to\infty}\barx(t_m)=x^*$ almost surely. From Proposition~\ref{prop:runav} it follows that  $\lim_{m\to\infty}x^{(j)}_i(t_m)=x^*$ a.s. for all $j\in[n_i]$, $i\in[n]$. Finally, by taking into account existence of the finite almost sure limit of $\|u(t)\|^2 = \sum_{i=1}^{n}\left(\frac{1}{n_i}\sum_{j=1}^{n_i}\|x^{(j)}_i(t)-x_i^*\|^2\right)$ as $t\to\infty$, we conclude that
  \[\Pr\{\lim_{t\to\infty}\|x^{(j)}_i(t)-x_i^*\|=0\}=1 \, \mbox{for all }j\in[n_i], \, i\in[n],\]
  and, therefore,
  \[\Pr\{\lim_{t\to\infty}\|x_i(t)-x_i^*\|=0\}=1 \, \mbox{for all } i\in[n].\]
\end{proof}

\section{Gradient Estimations}\label{sec:gradSampl}
In this section we present an approach to estimate the gradients of the agents' cost functions in such a way that Assumption~\ref{assum:step} is fulfilled.
The idea is borrowed from the work \cite{Bravo} dealing with bandit learning in games.
We assume the safety ball parameters $r_i$ and $p_i$ (see Assumption~\ref{assum:compact}) are known for each agent from the cluster $i$.
To obtain the estimation $d^j_i(t)$ based on the current estimation $x^{(j)}_i(t)$ and to follow the update in \eqref{eq:pbalg}, each agent $j\in[n_i]$ in the cluster $i$, $i\in[n]$, takes the following steps at time $t$.
The agent samples the vector $z^j_i(t)$ from the uniform distribution on the unit sphere $\Sph\subset \R^{n_i}$.
The query direction is defined by $w^{(j)}_i(t)=z^j_i(t) - r_i^{-1}(x^{(j)}_i(t) - p_i)$.
Then, the query point at which the oracle calculates the local cost function value is
\begin{align} \label{eq:query_point}
  \hat x^{(j)}_i(t)& = x^{(j)}_i(t) + \s_t w^{(j)}_i(t) \cr
  &=(1-\s_t r_i^{-1})x^{(j)}_i(t)+\s_t (z^j_i + r_i^{-1}p),
\end{align}
where $\s_t$ is the query radius chosen such that $\s_t r_i^{-1}<1$.
Note that, given  $x^{(j)}_i(t)\in\Om_i$, the query point $\hat x^{(j)}_i(t)$ above is feasible, i.e. $\hat x^{(j)}_i(t)\in \Om_i$ (see \cite{Bravo} for more details).
The gradient estimation itself is obtained as follows:
\begin{align}\label{eq:gr0}
 d^j_i(t) = \frac{n_i}{\s_t}J^{(j)}_i(\hat x^{(j)}_i(t), \hat{\tx}_{-i}(t))\cdot z^j_i,
 \end{align}
 where $\hat{\tx}_{-i}(t)$ is defined as in \eqref{eq:tx1}.
This vector is then used to follow the update in \eqref{eq:pbalg}.
As it has been proven in \cite{Bravo} (see, for example, (4.7) in \cite{Bravo}), $d^j_i(t)$, as constructed above, satisfies the following property:
\begin{align}\label{eq:grBravo}
 &d^j_i(t) = \nabla_iJ^j_i(x^{(j)}_i(t),\tx_{-i}(t)) + e^j_i(t),\cr
 &\mbox{where } \, \E_t\{\|e^j_i(t)\|\} = O(\s_t), \cr
 &\qquad\quad \E_t\{\|e^j_i(t)\|^2\} = O\left(\frac{1}{\s^2_t}\right)
\end{align}
with $\tx_{-i}(t)$ defined as in \eqref{eq:tx}.
Thus, for fulfillment of Assumption~\ref{assum:step} the step size parameter $\alpha_t$ and the query radius $\s_t$ must be balanced as follows:
\begin{align*}
  &\sum_{t=0}^{\infty}\alpha_t = \infty, \, \sum_{t=0}^{\infty}\alpha_t < \infty, \cr
  &\sum_{t=0}^{\infty}\alpha_t\s_t < \infty, \sum_{t=0}^{\infty}\frac{\alpha^2_t}{\s_t^2} < \infty.
\end{align*}
One example of an appropriate choice is $\alpha_t = \frac{\alpha_0}{t^a}$, $\s_t = \frac{\sigma_0}{t^b}$ with
\begin{align*}
 &\frac{1}{2}<a\le 1, \, b\ge 0, \\
 & a + b>1, \, 2a-2b>1.
\end{align*}
One possible parameter set is $a = 1$, $b=\frac{1}{3}$.
\begin{remark}
	There exist other approaches to one-point gradient estimations. The most known one corresponds to the queries sampled from the Gaussian distribution (see \cite{TatKam2019, Flaxman}).  However, to guarantee feasibility in this case, the query points have to be projected onto the action sets. To be able to control the deviation term, that appears due to this projection, one needs to introduce an auxiliary time-dependent parameter to the projection step of the main procedure (see \cite{TatKam2020, Flaxman}). Thus, introducing this parameter will somewhat complicate the analysis. That is why we leave the approach based on sampling from the Gaussian distribution beyond the scope of this paper.
\end{remark}

\section{Simulation Results} \label{sec:sim}
In this section, we verify our theoretical analysis with a practical simulation in order to show that the states of the agent system converge to the Nash equilibrium, defined in Definition \ref{def:NE}, when using the update equation \eqref{eq:pbalg} and the oracle gradient estimation of \eqref{eq:gr0}. As an example application, we chose a version of the well-known Cournot game. Consider the following setup: There are $n$ companies that compete against each other regarding the price of some specific product. Each company $i$ owns $n_i$ factories that produce said product. It is assumed that all factories produce the product with the same quality. The cost of factory $j$ belonging to cluster $i$ for producing the amount $x_i^j \in \mathbb{R}^+_0$ of the product is specific for this factory and defined by
\begin{equation}
	C_i^j(x_i^j ) = a_i^j (x_i^j)^2 + b_i^j x_i^j  + c_i^j.
\end{equation}
Naturally, the amount of product produced cannot be negative. Furthermore, each company has lower and upper production limits. The former defines a lower bound $\underline{x}_i^j$, for which production is still cost efficient, while the latter defines a production facility dependent upper bound $\overline{x}_i^j$.\\
Each company $i$ aggregates the product, produced in their $n_i$ factories, and sells it. In this version of the Cournot game, it is assumed that there exists only a single customer instance that buys all the aggregated product from all companies. The price that the customer pays per unit of product is dependent on the total supply by all companies and therefore defined as follows:
\begin{equation}
	P(x) = P_c - \sum_{i=1}^{n} \sum_{j=1}^{n_i} x_i^j,
\end{equation}
where $P_c$ is a constant, which is chosen such that for any decision vector $\underline{x} \leq x \leq \overline{x}$ it holds that $P(x) > 0$. With this price definition and assuming that the production costs for the factories belonging to the company is shared, each company $i$ aims to minimize its profit function, therefore solving the following optimization problem:
\begin{equation}\label{eq:cournotgame}
	\min_{\xlb_i \leq x_i \leq  \xub_i} \sum_{j=1}^{n_i} C_i^j(x_i^j ) - x_i^jP(x_i, x_{-i}) = \min_{\xlb_i \leq x_i \leq  \xub_i} J_i(x_i, x_{-i})
\end{equation}
It can be seen that the companies' profits are coupled by the customer's price function, therefore a Nash equilibrium needs to be found, from which no company has any incentive to deviate. Relating to the $n$-cluster games described in this paper, the companies represent the clusters, while the factories correspond to the agents. It can be readily confirmed that the profit optimization problem in Equation \eqref{eq:cournotgame} fulfils the Assumptions \ref{assum:convex} - \ref{assum:Lipschitz}.\\
\begin{table}
	\centering
	\begin{tabular}{ c|cccc|cccc}
		           & \multicolumn{4}{|c|}{Company $i = 1$} & \multicolumn{4}{c}{Company $i = 2$} \\ \hline
		$j$        & 1  & 2  & 3  &           4            & 1  & 2  & 3  &          4           \\ \hline
		$a_i^j$    & 5  & 8  & 4  &           5            & 3  & 7  & 9  &          2           \\ \hline
		$b_i^j$    & 10 & 11 & 9  &           12           & 10 & 11 & 12 &          9           \\ \hline
		$c_i^j$    & 1  & 3  & 2  &           5            & 3  & 2  & 3  &          1           \\ \hline
		$\xlb_i^j$ & 0  & 0  & 0  &           0            & 0  & 0  & 0  &          0           \\ \hline
		$\xub_i^j$ & 20 & 20 & 20 &           20           & 10 & 10 & 10 &          10
	\end{tabular}
	\label{tab:param}
	\caption{Parameters of agents in Cournot-Game.}
\end{table}
For our simulation, we choose a small setup consisting of two clusters, each containing four agents with parameters listed in Table \ref{tab:param} and cost coefficient $P_c = 250$. The agents inside the cluster $i$ are connected by an undirected communication graph $G_i$ that fulfils Assumption \ref{assum:connected}. Over this graph, state information is shared inside the cluster $i$ such that an estimation of all agent's states $x_i^{(j)}$ can be performed. The agents update their own gradient estimation according to equation \eqref{eq:gr0}, using the zero-order oracle information at the query point $\hat{x}_i^j$ defined in \eqref{eq:query_point}. The time-dependent, decreasing step-size $\alpha_t$ of the gradient update and the query radius $\sigma_t$ are chosen such that Assumption \ref{assum:step} is satisfied. Choosing a good set of parameters $\alpha_0, \sigma_0 >0$ and $a, b$ is crucial for the convergence speed of the algorithm. Even then, due to the fact that only zero-order information is available, convergence is slow. In Figure \ref{fig:states} the states, resulting from the application of the proposed algorithm to the scenario specified above, are plotted. The dashed line marks the true Nash Equilibrium $x^*$, while the two solid coloured lines distinguish between the states of company 1, i.e. $x_1$, and company 2, i.e. $x_2$, respectively.  Because there are only four agents in each cluster, which are connected by an almost fully connected graph, the consensus dynamic of the agent system is almost negligible against the gradient estimation and update dynamic. It can be seen that the algorithm converges to a satisfactory vicinity of the Nash Equilibrium states after about $1 \cdot 10^5$ Iterations at which the error norm between the agent's states and the true Nash equilibrium $x^*$ measures $||x - x^*||_2$ = 0.40. While the constraints of company 2 are not touched, the true Nash Equilibrium for two firms of the second company lies at their maximum production limit 10.

\begin{figure}
	\centering
	\includegraphics{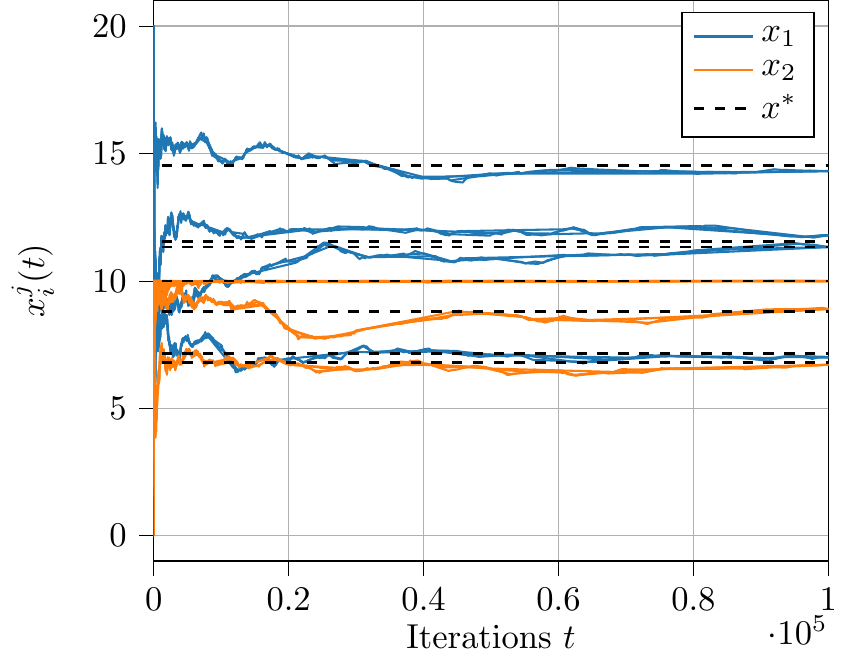}
	\label{fig:states}
	\caption{Convergence of the agent's states towards the Nash equilibrium of the n-cluster game. Error norm after $1 \times 10^5$ steps: $||x - x^*||_2 =  0.400$.}
\end{figure}

\section{Conclusion} \label{sec:conclusion}
In this paper we presented the distributed gradient play algorithm for strictly convex $n$-cluster games with communication setups within each cluster and a zero-order oracle in the whole system. We prove the almost sure convergence of this procedure to the unique Nash equilibria given an appropriate estimations of the local agents' gradients. The future work will be devoted to investigation of possible modifications which should enable a faster convergence rate. 

\appendix
The following is a well-known result of Robbins and Siegmund on non-negative random variables \cite{robbins1985convergence}.
\begin{theorem}\label{th:th_nonnegrv} Let $(\Omega, F, P)$ be a probability space and $F_1\subset F_2\subset\dots$ a sequence of sub-$\sigma$-algebras of $F$.
 Let $z_t, b_t, \xi_t,$ and $\zeta_t$ be non-negative $F_t$-measurable random variables satisfying
 \begin{align*}
  \E(z_{t+1}|F_t)\le z_t(1+b_t)+\xi_t-\zeta_t.
 \end{align*}
Then, almost surely $\lim_{t\to\infty} z_t$ exists and is finite for the case in which $\{\sum_{t=1}^{\infty}b_t<\infty, \;\sum_{t=1}^{\infty}\xi_t<\infty\}$. Moreover, in this case, $\sum_{t=1}^{\infty}\zeta_t<\infty$ almost surely.
\end{theorem}
\bibliographystyle{plain}

%\bibliography{document}

\begin{thebibliography}{10}

\bibitem{Bravo}
M.~Bravo, D.~Leslie, and P.~Mertikopoulos.
\newblock Bandit learning in concave n-person games.
\newblock In {\em Proceedings of the 32nd International Conference on Neural
  Information Processing Systems}, NIPS'18, page 5666–5676, Red Hook, NY,
  USA, 2018. Curran Associates Inc.

\bibitem{Flaxman}
A.D. Flaxman, A.T. Kalai, and H.B. McMahan.
\newblock Online convex optimization in the bandit setting: Gradient descent
  without a gradient.
\newblock In {\em Proceedings of the Sixteenth Annual ACM-SIAM Symposium on
  Discrete Algorithms}, SODA '05, pages 385--394, USA, 2005. Society for
  Industrial and Applied Mathematics.

\bibitem{Cortes2013}
B.~Gharesifard and J.~Cort\'es.
\newblock Distributed convergence to nash equilibria in two-network zero-sum
  games.
\newblock {\em Automatica}, 49(6):1683--1692, 2013.

\bibitem{Jarrah2015}
M.~Jarrah, M.~Jaradat, Y.~Jararweh, M.~Al-Ayyoub, and A.~Bousselham.
\newblock A hierarchical optimization model for energy data flow in smart grid
  power systems.
\newblock {\em Information Systems}, 53:190--200, 2015.

\bibitem{Meng2020}
M.~Meng and X.~Li.
\newblock On the linear convergence of distributed nash equilibrium seeking for
  multi-cluster games under partial-decision information.
\newblock {\em arXiv preprint:2005.06923}, 2020.

\bibitem{Nedich_Over}
A.~Nedi\'c and J.~Liu.
\newblock Distributed optimization for control.
\newblock {\em Annual Review of Control, Robotics, and Autonomous Systems},
  1(1):77--103, 2018.

\bibitem{Nedich_projected}
A.~{Nedi\'c}, A.~{Ozdaglar}, and P.~A. {Parrilo}.
\newblock Constrained consensus and optimization in multi-agent networks.
\newblock {\em IEEE Transactions on Automatic Control}, 55(4):922--938, 2010.

\bibitem{Niyato2011}
D.~{Niyato}, A.~V. {Vasilakos}, and Z.~{Kun}.
\newblock Resource and revenue sharing with coalition formation of cloud
  providers: Game theoretic approach.
\newblock In {\em 2011 11th IEEE/ACM International Symposium on Cluster, Cloud
  and Grid Computing}, pages 215--224, 2011.

\bibitem{FaccPang1}
J.-S. Pang and F.~Facchinei.
\newblock {\em Finite-dimensional variational inequalities and complementarity
  problems : vol. 1}.
\newblock Springer series in operations research. Springer, New York, Berlin,
  Heidelberg, 2003.

\bibitem{robbins1985convergence}
H.~Robbins and D.~Siegmund.
\newblock A convergence theorem for non negative almost supermartingales and
  some applications.
\newblock In {\em Herbert Robbins Selected Papers}, pages 111--135. Springer,
  1985.

\bibitem{BasharSG}
W.~Saad, H.~Zhu, H.~V. Poor, and T.~Ba{\c{s}}ar.
\newblock Game-theoretic methods for the smart grid: An overview of microgrid
  systems, demand-side management, and smart grid communications.
\newblock {\em IEEE Signal Processing Magazine}, 29(5):86--105, 2012.

\bibitem{Scutaricdma}
G.~Scutari, S.~Barbarossa, and D.~P. Palomar.
\newblock Potential games: A framework for vector power control problems with
  coupled constraints.
\newblock In {\em 2006 IEEE International Conference on Acoustics Speech and
  Signal Processing Proceedings}, volume~4, pages 241--244, May 2006.

\bibitem{TatKam2020}
M.~Kamgarpour T.~Tatarenko.
\newblock Bandit online learning of nash equilibria in monotone games.
\newblock {\em arXiv preprint:2009.04258}, 2020.

\bibitem{TatKam2019}
T.~{Tatarenko} and M.~{Kamgarpour}.
\newblock Learning nash equilibria in monotone games.
\newblock In {\em 2019 IEEE 58th Conference on Decision and Control (CDC)},
  pages 3104--3109, 2019.

\bibitem{elmark}
A.~C. Tellidou and A.~G. Bakirtzis.
\newblock Agent-based analysis of capacity withholding and tacit collusion in
  electricity markets.
\newblock {\em IEEE Transactions on Power Systems}, 22(4):1735--1742, Nov 2007.

\bibitem{Pang2020}
G.~Hu Y.~Pang.
\newblock Nash equilibrium seeking in n-coalition games via a gradient-free
  method.
\newblock {\em arXiv preprint:2008.12909}, 2020.

\bibitem{Ye2018}
M.~Ye, G.~Hu, and F.~L. Lewis.
\newblock {Nash equilibrium seeking for N-coalition noncooperative games}.
\newblock {\em Automatica}, 95:266--272, 2018.

\bibitem{Ye2020}
M.~Ye, G.~Hu, and S.~Xu.
\newblock {An extremum seeking-based approach for Nash equilibrium seeking in
  N-cluster noncooperative games}.
\newblock {\em Automatica}, 114:108815, 2020.

\bibitem{Zeng2019}
X.~Zeng, J.~Chen, S.~Liang, and Y.~Hong.
\newblock {Generalized Nash equilibrium seeking strategy for distributed
  nonsmooth multi-cluster game}.
\newblock {\em Automatica}, 103:20--26, 2019.

\bibitem{Zimmermann2021}
J.~Zimmermann, T.~Tatarenko, V.~Willert, and J.~Adamy.
\newblock Gradient-tracking over directed graphs for solving leaderless
  multi-cluster games.
\newblock {\em arXiv preprint:2102.09406}, 2021.

\end{thebibliography}

\end{document}